\documentclass[letterpaper, 10 pt, conference]{ieeeconf}  
\IEEEoverridecommandlockouts  
\usepackage{graphics} 
\usepackage{epsfig} 
\usepackage{mathptmx} 
\usepackage{times} 
\usepackage{amsmath} 
\usepackage{amssymb}

\usepackage{amsthm}
\usepackage{dsfont}
\usepackage{multirow}
\usepackage{color}
\usepackage{graphicx}
\usepackage[table,dvipsnames]{xcolor}

\usepackage[
    style=ieee,
    doi=false,
    isbn=false,
    url=true,
    eprint=false,
    backend=bibtex,
    natbib=true
    ]{biblatex}

\newtheorem{thm}{Theorem}
\newtheorem{lm}{Lemma}

\newcommand{\remove}[1]{#1}

\newcommand{\add}[1]{{\iffalse #1 \fi}} 

\bibliography{mendeley_v2,refs}

\title{\LARGE \bf
Simulation and Real-World Evaluation of Attack Detection Schemes
}

\author{Matthew Porter$^{1}$, Arnav Joshi$^{1}$, Pedro Hespanhol$^{2}$,\\ Anil Aswani$^{2}$, Matthew Johnson-Roberson$^{3}$, and Ram Vasudevan$^{1}$
\thanks{*This work was supported by a grant from Ford Motor Company via the Ford-UM Alliance under award N022977.}
\thanks{$^{1}$Matthew Porter, Arnav Joshi, and Ram Vasudevan are with the Department of Mechanical Engineering,
        University of Michigan, Ann Arbor, MI 48103, USA 
        {\tt\small \{matthepo,arnavj,ramv\}@umich.edu}.}%
\thanks{$^{2}$Pedro Hespanhol and Anil Aswani are with the Department of Industrial Engineering and Operations Research, 
        University of California Berkeley, Berkeley, CA 94720, USA 
        {\tt\small \{pedrohespanhol,aaswani\}@berkeley.edu}.}%
\thanks{$^{3}$Matthew Johnson-Roberson is with the Department of Naval Architecture, 
        University of Michigan, Ann Arbor, MI 48103, USA 
        {\tt\small mattjr@umich.edu}.}%
}

\begin{document}
\maketitle
\thispagestyle{empty}
\pagestyle{plain}

\begin{abstract}
A variety of anomaly detection schemes have been proposed to detect malicious attacks to Cyber-Physical Systems.
Among these schemes, Dynamic Watermarking methods have been proven highly effective at detecting a wide range of attacks. 
Unfortunately, in contrast to other anomaly detectors, no method has been presented to design a Dynamic Watermarking detector to achieve a user-specified false alarm rate, or subsequently evaluate the capabilities of an attacker under such a selection. 
This paper describes methods to measure the capability of an attacker, to numerically approximate this metric, and to design a Dynamic Watermarking detector that can achieve a user-specified rate of false alarms. 
The performance of the Dynamic Watermarking detector is compared to three classical anomaly detectors in simulation and on a real-world platform.
These experiments illustrate that the attack capability under the Dynamic Watermarking detector is comparable to those of classic anomaly detectors.
Importantly, these experiments also make clear that the Dynamic Watermarking detector is consistently able to detect attacks that the other class of detectors are unable to identify.
\end{abstract}

\section{Introduction}
Cyber-Physical Systems (CPS) have proven difficult to secure due to the ever-present risk of malicious attacks.
Failure to detect such attacks can have catastrophic consequences \cite{Abrams2008,Langner2011, Lee2016AnalysisGrid}. 
Though these real-world attacks, such as the Stuxnet Worm or the attack on the Ukranian Power Grid, highlight the threat to existing industrial facilities and public utilities, researchers speculate that attacks on next generation transportation systems could be even more dangerous due to the reliance on communication between infrastructure and privately owned vehicles \cite{Amoozadeh2015,Dominic2016,Chen2018ExposingControl}.

To counter the growing risks of attacks on CPS, researchers have attempted to develop techniques to detect attacks while they are being conducted.
Rather than address all possible attacks, researchers have focused on detecting additive attacks in which an attacker alters a measurement that is used while performing feedback control.
To identify these attacks at run-time, detection schemes compute the residual between a state observer and measurements and then analyze this residual signal to determine whether an attack is taking place.
Three detectors, which were originally proposed for anomaly detection in quality control applications \cite{Hotelling1947MultivariateControl}, analyze the residual signal:
the $\chi^2$, cumulative sum (CUSUM), and multivariate exponentially weighted moving average (MEWMA) detectors \cite{Mo2010FalseSystems,Murguia2016CUSUMSensors}.

\begin{figure}[t]
    \centering
    \includegraphics[trim={0.2in 0.5in 0.2in 0.05in},clip,width=0.4\textwidth]{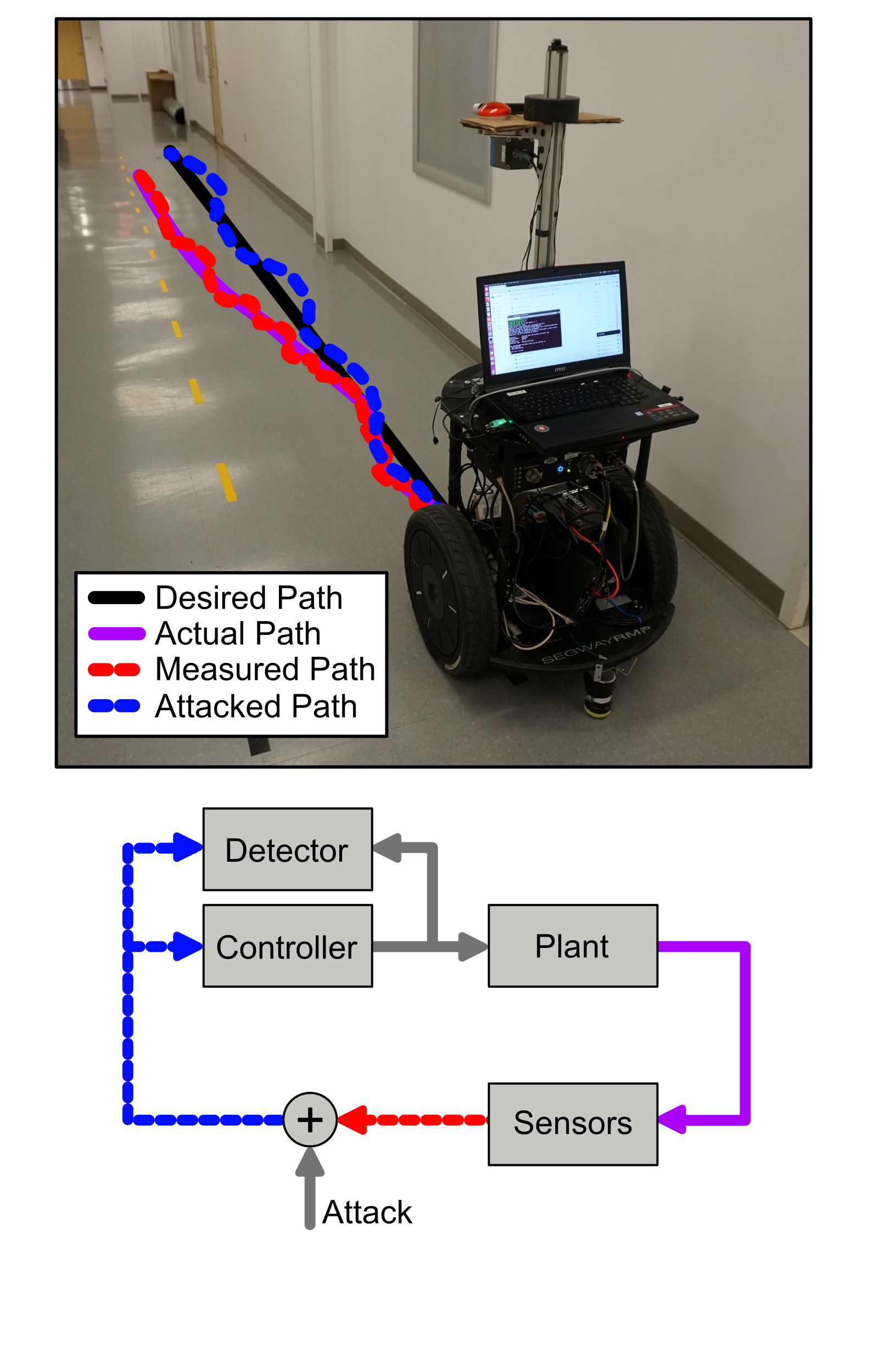}
    \caption{This paper describes a real-world implementation of malicious attack detection on a Segway Robotics Mobility Platform performing a path-following task (top).
    An additive attack is applied to this system as shown in the block diagram (bottom).
    Four different detectors ($\chi^2$, CUSUM, MEWMA, and the presented Dynamic Watermarking method) are applied in the ``Detector'' block.
    The Dynamic Watermarking method is shown to detect attacks that the others cannot, without sacrificing tracking performance relative to the desired path.
    Video of this example can be found at \cite{Porter_vid_2018}.}
    \label{fig:overview}
    \vspace*{-1.25em}
\end{figure}

To detect an attack, each of the detectors evaluates a test statistic that is a function of the residual.
If this statistics value rises above some user-specified threshold, then the detector triggers an alarm.
To evaluate and design the threshold for these detectors, researchers have proposed the following three metrics: first, the \emph{attack capability} or the amount of perturbation to the state of the system that an attack can induce without either inducing an alarm \cite{Murguia2016CUSUMSensors} or without increasing the rate of alarms \cite{Murguia2018OnAttacks,Mo2016}; second, the rate of false alarms ($R_{FA}$) given by the detector when no attack is occurring; and third, the ability of the detector to reliably detect specific attack models. 
For an open-loop stable system, the attack capability can be evaluated by computing the reachable set of the error in the observed state. 
Since computing this reachable set can be challenging, researchers have instead attempted to evaluate surrogates for the attack capability such as the expected value of the state vector \cite{Murguia2016CUSUMSensors} or the norm of the largest time invariant residual \cite{Umsonst2018}.
However, these surrogates are unable to accurately characterize the attack capability of attacks that have large residuals for short amounts of time. 
Note that, by reducing the threshold in any detector, one can reduce the attack capability; however, this can increase the rate of false alarms.
To compute this false alarm rate for classic anomaly detectors, it is typically assumed that the residuals are independent \cite{Murguia2016CUSUMSensors,Umsonst2018}. 
Unfortunately, simulated experiments have noted that this assumption can lead to a consistent error while computing the false alarm rate \cite{Murguia2016CUSUMSensors}. 
Researchers have also shown that a $\chi^2$ detector is always capable of identifying a specific type of additive attack where a measured signal is replayed \cite{Mo2009}; however, the ability to detect specific attacks has been less studied for the other detectors. 

More recently, researchers have begun exploring techniques to detect more sophisticated attacks, which exploit knowledge of the system dynamics.
These detectors rely on \emph{Dynamic Watermarking} wherein an excitation signal that is only known to the control system is introduced into the input.
The detector then evaluates the covariance of the residual signal with the watermark to determine whether the system is under attack.
These Dynamic Watermarking based detectors are theoretically proven to detect attacks that exploit knowledge of the system dynamics.
Initially, Dynamic Watermarking was developed for LTI systems with full rank input matrices and full state observations. 
These methods were proven to detect arbitrary attack models \cite{Satchidanandan2017} including attacks that replayed a measured signal \cite{Weerakkody2014}. 
These methods were later extended to generalized LTI systems \cite{Hespanhol2017},
and to networked control systems \cite{Hespanhol2018}.

Though Dynamic Watermarking is proven to be capable of detecting a larger class of attacks when compared to prior detection algorithms, to the best of our knowledge, no one has conducted a real-world evaluation of any of these attack detection algorithms. 
Moreover, no one has evaluated the attack capability of a system employing a Dynamic Watermarking scheme as a function of false alarm rate or developed a technique to design a detector using Dynamic Watermarking that achieves a user-specified false alarm rate. 

The contributions of this paper are three-fold.
First, in Section \ref{sec:reach}, we construct a metric for quantifying the attack capability based on the reachable set of the portion of the observer error related to the residual without triggering an alarm, and we develop a technique to compute an outer approximation of this reachable set for the $\chi^2$, CUSUM, MEWMA, and Dynamic Watermarking detectors.
Second, in Section \ref{sec:false}, we develop an empirical method to design a Dynamic Watermarking detector that achieves a user-specified false alarm rate.
Third, in Section \ref{sec:simp}, the attack capability for each detector is compared, and in Section \ref{sec:exp} a real world example, shown in Figure \ref{fig:overview}, is presented to evaluate the performance of the $\chi^2$, CUSUM, MEWMA, and Dynamic Watermarking detectors for specific attacks. 
See \cite{Porter_vid_2018} for a video of this real-world test.
The rest of this document is outlined as follows.
In Section \ref{sec:prelim}, we define the notation used in the paper, along with the LTI system model and assumptions.
In Section \ref{sec:alg} we provide notation for each of the detectors.
We draw conclusions in Section \ref{sec:con}.

\section{Preliminaries}
\label{sec:prelim}

This section describes the notation and assumptions used throughout the paper for the LTI system model. 

\subsection{Notation}
\label{subsec:notation}
This paper uses three different probability distributions:
the multivariate Gaussian distribution, denoted $\mathcal{N}(\mu,\Sigma)$, with mean $\mu$ and covariance $\Sigma$;
the $\chi^2$ distribution, denoted $\chi^2(i)$, with $i$ degrees of freedom; and the the Wishart distribution, denoted $\mathcal{W}(\Sigma,i)$, with scale matrix $\Sigma$ and $i$ degrees of freedom \cite[Section 7.2]{Anderson2003AnAnalysis}.
The Euclidean Norm of a vector $X \in \mathbb{R}^n$ is denoted $\|X\|$. 
Similarly, the induced operator norm for a matrix $X \in \mathbb{R}^{n \times m}$  is denoted $\|X\|$. 
The expectation of a variable $X$ is denoted $\mathds{E}[X]$.
Zero matrices of dimension $i\times j$ are denoted $0_{i\times j}$, and in the case that $i=j$ the notation is simplified to $0_i$. 
The identity matrix of dimension $i$ is denoted $I_i$. 
The closed unit ball of radius $\epsilon$ is denoted $\mathcal{B}_\epsilon$. 
The Minkowsi sum is denoted $\oplus$. 
The minimum singular value of a matrix $X \in \mathbb{R}^{n\times m}$ is denoted $s_1(X)$. 

\subsection{LTI Model}
This paper considers a discrete linear time invariant systems where the discrete time step is indexed by $n \in \mathbb{N}$:
\begin{align}
    x_{n+1} &= Ax_n+Bu_n+w_n\\
    y_n &= Cx_n+z_n+v_n
\end{align}
with state $x_n\in\mathbb{R}^p$, measurement $y_n\in\mathbb{R}^q$, and input $u_n\in\mathbb{R}^m$. 
The process noise $w_n\in\mathbb{R}^p$ and the measurement noise $z_n\in\mathbb{R}^q$ are assumed to be Gaussian with $w_n\sim\mathcal{N}(0,\Sigma_w)$ and $z_n\sim\mathcal{N}(0,\Sigma_z)$. 
At each time step, the attacker adds $v_n\in\mathbb{R}^q$ to the measurement.
For each $n \in \mathbb{N}$ an observer recovers the full observed state $\hat{x}_n$ from the measurements:  
\begin{align}
    \hat{x}_{n+1}=&(A+LC)\hat{x}_n+Bu_n-Ly_n.
\end{align}
The observed state is then used in full state feedback:
\begin{align}
    u_n=&K\hat{x}_n+e_n
\end{align}
where $e_n\sim\mathcal{N}(0,\Sigma_e)$ is a private watermark and is only added when running the Dynamic Watermarking detector. 
The controller gain matrix $K$ and the observer gain matrix $L$ are chosen such that the closed loop matrices $(A+BK)$ and $(A+LC)$ are Schur Stable.
For each $n \in \mathbb{N}$, the observer error $\delta_{n}=\hat{x}_n-x_n$ then has the following update equation:
\begin{equation}
    \delta_{n+1}=(A+LC)\delta_n-w_n-Lz_n-Lv_n.\label{eq:observe_error}
\end{equation}
Denote the residual as $r_n$ which is defined as:
\begin{equation}
    r_n=C\hat{x}_n-y_n=C\delta_n-z_n-v_n.\label{eq:res}
\end{equation}
When the system is not being attacked, i.e. $v_n=0,~\forall n\in \mathbb{N}$, the steady state covariance for the observer error $\Sigma_\delta=\lim_{n\rightarrow \infty}\mathds{E}[\delta_n\delta_n^T]$ can be found by a Discrete Lyapunov Equation as described in \cite[Section 2.2.2]{Lewis2007}:
\begin{align}
    \Sigma_\delta=&(A+LC)\Sigma_\delta(A+LC)^T+\Sigma_w+L\Sigma_zL^T.\label{eq:lyap}
\end{align}
We can then use the steady state covariance of the observer error to find the steady state covariance of the residual $\Sigma_r$ when the system is not being attacked:
\begin{align}
    \Sigma_r&=\lim_{n\rightarrow \infty}\mathds{E}[r_nr_n^T]\\
    &=\lim_{n\rightarrow \infty}\mathds{E}[C\delta_n\delta_n^TC^T-C\delta_nz_n^T-z_n\delta_n^TC^T+z_nz_n^T]\label{eq:cov_cross}\\
    &=C\Sigma_\delta C^T+\Sigma_z,
\end{align}
where the expectation of the cross terms in \eqref{eq:cov_cross} are zero due to causality.
Using the steady state covariance of the residuals, define the normalized residual $\bar{r}_n$:
\begin{align}
    \bar{r}_n=\Sigma_r^{-1/2}r_n.\label{eq:norm_res}
\end{align}
Furthermore, we define the vector of current and previous normalized residuals:
\begin{align}
R_n&=[\bar{r}_n^T\hdots \bar{r}_0^T]^T.\label{eq:R}
\end{align}

\section{Detection Algorithms}
\label{sec:alg}

This section describes the technical details of each of the detection algorithms explored in this paper. 
At each time step, each detector computes a test statistic, $a_{*}(R_n)$, based on current and previous residuals. 
The subscript $*$ is a placeholder for the detector designation, and $n$ is the discrete time step. 
If, for a given detector, the test statistic exceeds a predefined threshold, then the detector raises an alarm. 
These thresholds are denoted by $\tau_*$.

\subsection{$\chi^2$ Detector}
The $\chi^2$ detector uses the normalized residual $\bar{r}_n$ to develop a statistical test. 
Since, under the assumption of no attack, $\bar{r}_n\sim\mathcal{N}(0,I)$, the \emph{$\chi^2$ detector} is defined as:
\begin{align}
a_{\chi^2}(R_n)=\bar{r}_n^T\bar{r}_n<\tau_{\chi^2},
\end{align}
where the test statistic $a_{\chi^2}(R_n)\sim\chi^2(q)$ when the system is not under attack. 
A change to the distribution of the residual, as a result of an attack, may change the resulting distribution of the $\chi^2$ test value, but this is not true for all attacks. 
For instance an attack could replace the residual vectors with:
\begin{align}
    r^\prime_n=\Sigma_r^{1/2}\bar{r}_n^\prime=\Sigma_r^{1/2}\begin{bmatrix}\sqrt{c_n}&0&\hdots&0\end{bmatrix}^T
\end{align}
where $c_n\sim\chi^2(q)$.
Then $\bar{r}_n^{\prime T}\bar{r}_n^\prime=c_n\sim\chi^2(q)$.
Similarly some attacks, such as the one described in Section \ref{sec:exp}, can generate a false set of residuals that have the same distribution as the residual when no attack is taking place. 
Such attacks would be indistinguishable by the $\chi^2$ detector while increasing the error in the observed state.
Furthermore, the $\chi^2$ detector is memoryless, which can make detecting small increases in the norm of the residuals difficult.

\subsection{CUSUM Detector}
The CUSUM detector addresses the difficulty in detecting small but persistent increases in the norm of the normalized residual by introducing dynamics to its test statistic:
\begin{align}
a_{C}(R_n)=\max(a_{C}(R_{n-1})+\bar{r}_n^T\bar{r}_n-\gamma,0)<\tau_C,
\end{align}
where $a_{C}(R_{-1})=0$, and $\gamma$ is a parameter called the \emph{forgetting factor}. 
To ensure that the test statistic is stable, $\gamma>q$ where $q$ is the dimension of the residual \cite[Theorem 1]{Murguia2016CUSUMSensors}.
Similar to the $\chi^2$ detector, the CUSUM detector bounds the norm of the normalized residual under the assumption of no alarms:
\begin{align}
    \bar{r}_n^T\bar{r}_n-\gamma\leq a_{C}(R_{n-1})+\bar{r}_n^T\bar{r}_n-\gamma<\tau_C.
\end{align}
For the $CUSUM$ detector, a persistent increase in the norm of the normalized residual, increases the likelihood that  $\bar{r}_n^T\bar{r}_n^T>\gamma$.  
When this is true for several steps, the CUSUM test value increases cumulatively, triggering an alarm.

\subsection{MEWMA Detector}
The MEWMA detector test statistic also incorporates dynamics. 
The MEWMA detector uses the exponentially weighted moving average of the normalized residual:
\begin{align}
    G_{n}=&\beta\bar{r}_n+(1-\beta)G_{n-1}
\end{align}
where $G_{-1}=0$ and the parameter $\beta\in(0,1]$ is also called the \emph{forgetting factor}.
The MEWMA detector is then defined as:
\begin{align}
    a_{M}(R_n)=&\frac{2-\beta}{\beta}G_{n}^TG_n<\tau_M.
\end{align}
When $\beta=1$ the test statistic is equal to the $\chi^2$ detector's test statistic.
For smaller $\beta$, one gets a similar effect to that of the CUSUM detector, because, for a forgetting factor $\beta\in(0,1)$, a persistent increase in the norm of the residual results in a larger value of $G$ which results in a higher test statistic value.
However, the MEWMA test statistic does not increase for all persistent changes.
For instance, if the covariance of the residuals under attack are $\alpha\Sigma_r$ for some $\alpha\in(0,1)$, we expect a lower test value.

\subsection{Dynamic Watermarking}

Dynamic Watermarking is designed to address the shortcomings of the previous detection algorithms by sounding an alarm not only for changes in the distribution of the norm of the normalized residual (as the $\chi^2$ and CUSUM detectors are able to do), but also for persistent changes in the distribution of the normalized residual as the MEWMA detector is able to do. 
To illustrate this, we now briefly summarize the results of \cite{Hespanhol2017}. 
In its statistical limit form, Dynamic Watermarking was developed to detect any persistent change made to the residual:

\begin{thm}\label{thm:asymp_attack_conv_to_0_hespanhol_thm1}\cite[Theorem 1] {Hespanhol2017}
Suppose (A,B) is stabilizable, (A,C) is detectable, $\Sigma_e$ is full rank,  and: 
\begin{align}
    k^\prime=\min\{k\geq0~|~C(A+BK)^kB\neq0\}\label{eq:k}.
\end{align} 
If: 
\begin{align}
    \underset{N\rightarrow\infty}{aslim}\frac{1}{N}\sum_{i=0}^{N-1}r_nr_n^T&=\Sigma_r\quad \text{and} \label{eq:cond1}\\
    \underset{N\rightarrow\infty}{aslim}\frac{1}{N}\sum_{i=0}^{N-1}r_ne_{n-k}^T&=0,\label{eq:cond2}
\end{align}
then the asymptotic attack power, defined as:
\begin{align}
    \underset{N\rightarrow\infty}{\lim}\frac{1}{N}\sum_{i=0}^{N-1}v_n^Tv_n
\end{align}
must converge to 0. 
\end{thm}
For an attack to have a persistent effect on the system, the asymptotic attack power must be greater than 0. 
Since the test described in Theorem \ref{thm:asymp_attack_conv_to_0_hespanhol_thm1} requires evaluating an infinite sum, it does not consider attacks that are not persistent for all time. 
To detect attacks which are not persistent in time, the infinite time tests are transformed into finite time tests by taking a sliding window of the combined values of the sums in \eqref{eq:cond1} and \eqref{eq:cond2}:
\begin{align}
    D_n=&\sum_{n-\ell+1}^n\psi_n\psi_n^T\\
    \psi_n^T=\Sigma_\psi^{-1/2}&\begin{bmatrix}r_n^T & (e_{n-k})^T\end{bmatrix}\label{eq:normalize}
\end{align}
where $\ell$ is the size of the sliding window and:
\begin{align}
    \Sigma_\psi=\begin{bmatrix}\Sigma_r & 0\\0 & \Sigma_e\end{bmatrix}.
\end{align}
In this paper, we have normalized the vector $\psi_n$ in \eqref{eq:normalize} for the convenience of later derivations.
The matrix quantity $D_n$ is then approximately distributed according to a Wishart Distribution, $\mathcal{W}(I,\ell)$ \cite{Hespanhol2017}, giving rise to the Dynamic Watermarking detector we use in this work:
\begin{align}
    a_{\mathcal{D}}(R_n)=&\mathcal{L}_{m+q}^\ell(D_n)<\tau_\mathcal{D}
\end{align}
where $\mathcal{L}$ is the negative log likelihood function:
\begin{align}
\mathcal{L}_{i}^j(X)=&\frac{(i+1-j)}{2}\cdot\log(|X|)+\frac{1}{2}\text{trace}\left(X\right)+\nonumber\\
    &+\log\left(2^{ij/2}\Gamma_{(i)}(\frac{j}{2})\right).\label{eq:lik}
\end{align}
Note, we have also modified the formula from \cite{Hespanhol2017}, by including a constant term, and removing the scale matrix inverse from the trace in \eqref{eq:lik} due to the normalization carried out in \eqref{eq:normalize}. 
Note that for $n<\ell+k^\prime$, the test statistic value is defined as $0$.
\section{Analytical Comparisons}
\label{sec:ana}



This section describes the metric for the attack capability, and derives methods for approximating the attack capability and the $R_{FA}$ for the $\chi^2$, CUSUM, MEWMA, and Dynamic Watermarking detectors.\add{ Due to space considerations proofs for the Theorems and Lemmas in this section can be found in the technical report \cite{Porter2018}.}


\subsection{Attack Capability}
\label{sec:reach}
%
Assuming the $A$ matrix is Schur Stable, the capability of an attack can be measured by its ability to affect the observer error $\delta_n$.
A reachable set of the observer error can evaluate the attack capability, but, since the noise is supported over an infinitely large set, this reachable set would have infinite volume. 
As a result, this work focuses on computing the volume of the reachable set of the portion of the observer error corresponding to the residual under the condition of no alarms being raised. 
To provide a rigorous definition of this set, we introduce some additional notation and definitions.

Using superposition, one can split the observer error described in \eqref{eq:observe_error} into two pieces:
\begin{align}
\delta^{(a)}_{n+1}&=(A+LC)\delta_n^{(a)}-Lz_n-Lv_n\label{eq:observer_sep}\\
\delta^{(b)}_{n+1}&=(A+LC)\delta_n^{(b)}-w_n.
\end{align}
The observer error is then $\delta_n=\delta_n^{(a)}+\delta_n^{(b)}$. 
Here $\delta_n^{(a)}$ is the portion related to the residual, which can be seen by applying \eqref{eq:res} and \eqref{eq:norm_res} to \eqref{eq:observer_sep}:
\begin{align}
    \delta^{(a)}_{n+1}&=A\delta^{(a)}_n+L\Sigma_r^{1/2}\bar{r}_n.\label{eq:errdyn2}
\end{align}
Since an attack is only able to affect the $\delta_n^{(a)}$ portion of the observer error, the other portion is ignored while evaluating attack capability.
For each $n\in\mathbb{N}$, denote the \emph{reachable set of $\delta_n^{(a)}$ at a given time step $n$ under the condition of no alarms for a threshold $\tau_*$} as $\mathcal{R}_n^{\tau_*}$ and define it as:
\begin{align}
\mathcal{R}_n^{\tau_*}=\{\delta_n^{(a)}~|~\delta^{(a)}_n=\bar{A}_{n-1}R_{n-1},~R_{n-1}\in\Omega_n^{\tau_*}\}\label{eq:Rnset}
\end{align}
where:
\begin{align}
    \Omega_{n-1}^{\tau_*}=\{R_{n-1}~|~ a_{*}(R_{n-1})<\tau_*~\forall i<n\}\label{eq:omega},
\end{align}
and:
\begin{align}
    \bar{A}_n&=\begin{bmatrix}L\Sigma_r^{1/2}&AL\Sigma_r^{1/2}\hdots~ A^{n}L\Sigma_r^{1/2}\end{bmatrix}.\label{eq:abar}
\end{align}
Furthermore, we denote the \emph{steady state reachable set under the condition of no alarms for a threshold $\tau_*$} as $\mathcal{R}^{\tau_*}$ and define it as:
\begin{align}
\mathcal{R}^{\tau_*}=\{\delta^{(a)}~|~\forall n\in\mathbb{N},~\exists m\in\mathbb{N} \text{ s.t.} ~m>n, \delta^{(a)}\in\mathcal{R}_m^{\tau_*} \}.\label{eq:Rset}
\end{align}
Finally, we evaluate the attack capability by measuring the volume of the steady state reachable set under the condition of no alarms under a threshold $\tau_*$, which is defined as:
\begin{align}\label{eq:VRS}
    V_{RS}(\tau_*)=\mu(\mathcal{R}^{\tau_*}),
\end{align}
where $\mu$ denotes the Lebesgue measure.

Calculating the set $\mathcal{R}^{\tau_*}$ can be difficult, so we first derive a method for calculating $\mathcal{R}^{\tau_*}_n$:
\begin{thm}
Suppose $\tau_*\in\mathbb{R}$ and $\bar{A}_{n-1},~R_{n-1}$, $\mathcal{R}_n^{\tau_*}$, and $\Omega_{n-1}^{\tau_*}$ are as in \eqref{eq:abar}, \eqref{eq:R}, \eqref{eq:Rnset}, and \eqref{eq:omega}, respectively. 
Suppose $v:\mathbb{R}^q\rightarrow \mathbb{R}$ is the solution to:
\begin{flalign}
    & & \underset{v\in \mathcal{C}}{\text{inf}}\hspace*{0.25cm} & \int v(\delta) ~d\delta &&\\
    & & \text{s.t.}\hspace*{0.25cm} & v(\delta)\geq 0 && \delta\in\mathbb{R}^q\\
    & & & v(\bar{A}_{n-1}R_{n-1})-1\geq 0 && R_{n-1}\in \Omega^{\tau_*}_{n-1}\label{eq:const}
\end{flalign}
where $\mathcal{C}$ is the space of continuous functions.
Then the 1 super-level set of $v$ is an outer approximation to $\mathcal{R}_n^{\tau_*}$
\end{thm}
\remove{\begin{proof}
Let $\delta\in\mathcal{R}_n^{\tau_*}$.
Then, from \eqref{eq:Rnset}, there exists an $R_{n-1}\in\Omega_{n-1}^{\tau_*}$ such that $\delta=\bar{A}_{n-1}R_{n-1}$. 
The constraint in \eqref{eq:const} then gives $v(\bar{A}_{n-1}R_{n-1})=v(\delta)\geq 1$.
\end{proof}}

To make this problem computationally tractable, we optimize over polynomial functions of fixed degree instead of continuous functions, and describe the positivity constraint, \eqref{eq:const}, with a Sums-of-Squares constraint. 
We then apply Sums-of-Squares Programming to generate an outer approximation to the reachable set. 
To replace \eqref{eq:const} with a Sums of Squares constraint, $\Omega_n^{\tau_*}$ must first be replaced with a semi-algebraic set \cite[Theorem 2.14]{Lasserre2010MomentsApplications}.
To simplify our exposition, we denote by $\Theta_n^{\tau_*}$ a collection of semi-algebraic constraints such that $\Omega_{n}^{\tau_*}\subseteq\Theta_n^{\tau_*}$.
In fact, as we show next, for many detectors, $\Omega_{n}^{\tau_*} = \Theta_n^{\tau_*}$.

For the $\chi^2$ detector, the constraint of no alarms is a quadratic constraint on the residual, so:
\begin{align}
\Theta_n^{\tau_{\chi^2}} = \left\{R_n~|~R_n^TQ_{(i,n)}^{\tau_{\chi^2}}R_n<1~i=0,...,n\right\}\label{eq:theta1}
\end{align}
where:
\begin{align}\label{eq:chi_con_set_quad2}
    Q_{(i,n)}^{\tau_{\chi^2}}&=\frac{1}{\tau_{\chi^2}}\begin{bmatrix}0_{q(n-i)}&0&0\\0&I_{q}&0\\0&0&0_{q(i)}\end{bmatrix}.
\end{align}
Note that $\Theta_n^{\tau_{\chi^2}}=\Omega_n^{\tau_{\chi^2}}$ since $\frac{a_{\chi^2}(R_i)}{\tau_{\chi^2}}=R_n^TQ_{(i,n)}^{\tau_{\chi^2}}R_n$ for all $i\leq n$.

For the CUSUM detector:
\begin{align}
\Theta_n^{\tau_{C}} = \left\{R_n~|~R_n^TQ_{(i,j,n)}^{\tau_{C}}R_n<1~i=0,...,n~j\leq i\right\}\label{eq:theta2}
\end{align}
where:
\begin{align}\label{eq:CUSUM_con_set_quad2}
Q_{(i,j,n)}^{\tau_C}&=\frac{1}{\tau_C+\gamma (j+1)}\begin{bmatrix}0_{q(n-i)}&0&0\\0&I_{q(j+1)}&0\\0&0&0_{q(i-j)}\end{bmatrix}.
\end{align}
Note that $\Theta_n^{\tau_{C}}=\Omega_n^{\tau_{C}}$, since:
\begin{align}
    a_{C}(R_i)=\max\left(\left\{\sum_{h=i-j}^i(\bar{r}_h^T\bar{r}_h-\gamma)~|~j\leq i\right\},0\right)
\end{align}
and:
\begin{align}
    R_n^TQ_{(i,j,n)}R_n=\frac{1}{\tau_C+\gamma (j+1)}\sum_{h=i-j}^{i}\bar{r_h}^T\bar{r}_h<1
\end{align}
can be rearranged to form:
\begin{align}
    \sum_{h=i-j}^i(\bar{r}_h^T\bar{r}_h-\gamma)<\tau_C.
\end{align}

For the MEWMA detector note that:
\begin{align}
\Theta_n^{\tau_M} = \left\{R_n~|~R_n^TQ_{(i,n)}^{\tau_M}R_n<1~i=0,...n\right\}\label{eq:theta3}
\end{align}
where:
\begin{align}\label{eq:MEWMA_con_set_quad2}
Q_{(i,n)}^{\tau_M} &= \frac{2-\beta}{\beta\tau_M}\begin{bmatrix}0_{q(n-i)\times q}\\\beta I_{q}\\(1-\beta)\beta I_q\\\vdots\\(1-\beta)^{i}\beta I_q\end{bmatrix}\begin{bmatrix}0_{q(n-i)\times q}\\\beta I_q\\(1-\beta)\beta I_q\\\vdots\\(1-\beta)^{i}\beta I_q\end{bmatrix}^T.
\end{align}
Note that $\Theta_n^{\tau_{M}}=\Omega_n^{\tau_{M}}$ since $\frac{a_{M}(R_i)}{\tau_{M}}=R_n^TQ_{(i,n)}^{\tau_{M}}R_n$ for all $i\leq n$.

While $\Omega_{n}^{\tau_*}$ is already a semi-algebraic set for the $\chi^2$, CUSUM and the MEWMA detectors, this is not true for the Dynamic Watermarking detector due to the log function in \eqref{eq:lik}. 
Therefore, we consider an outer approximation to $\Omega_n^{\tau_\mathcal{D}}$ described via a quadratic constraint:
\begin{thm}\label{thm:SHREYAS_LABELED_THIS}
Suppose $\Omega_n^{\tau_\mathcal{D}}$ is as in \eqref{eq:omega}, $\ell$ is the window size of the Dynamic Watermarking detector, $\tau_\mathcal{D}$ is the threshold of the detector, $k^\prime$ is as in \eqref{eq:k}, $q$ is the dimension of the residual, $m$ is the dimension of the input signal, and:
\begin{align}
    \Theta_n^{\tau_\mathcal{D}} = \left\{R_n~|~R_n^TQ_{(i,n)}^{\tau_\mathcal{D}}R_n<1~i=\ell+k^\prime,...,n\right\},\label{eq:theta4}
\end{align}
where:
\begin{align}\label{eq:DW_con_set_quad2}
    Q_{(i,n)}^{\tau_\mathcal{D}}=\frac{1}{(m+q)\epsilon}\begin{bmatrix}0_{q(n-i)} & 0 & 0\\0 & I_{q\ell} & 0\\0 & 0 & 0_{q(i-\ell)}\end{bmatrix}
\end{align}
and where $\epsilon>\ell-1-q-m$ is a solution to:
\begin{align}    
    \tau_\mathcal{D}=\frac{(q+m)\epsilon}{2} &+\frac{(q+m+1-\ell)}{2}\log(\epsilon^{q+m}) + \nonumber \\&+\log\left(2^{(q+m)\ell/2}\Gamma_{(q+m)}(\frac{\ell}{2})\right).\label{eq:epsilonD}
\end{align}
Then $\Omega_n^{\tau_\mathcal{D}} \subset \Theta_n^{\tau_\mathcal{D}}$.
\end{thm}
\remove{To prove this theorem, consider the following lemma:
\begin{lm}
\label{lm:convex}
Suppose $(g_i){i \in \mathbb{N}}$ is a sequence of vectors where $g_i\in\mathbb{R}^q$, $\tau\in\mathbb{R}$ such that $\tau>0$, and $\ell\in\mathbb{N}$ such that $\ell>q+1$. Furthermore suppose that:
\begin{align}
    \mathcal{L}_q^\ell(\sum_{i=1}^\ell g_ig_i^T)<\tau
\end{align}
where the function $\mathcal{L}_q^\ell$ is as in \eqref{eq:lik}. Then:
\begin{align}
\sum_{i=1}^\ell g_i^Tg_i<\epsilon q,
\end{align}
where $\epsilon>\ell-1-q$ is a solution to:
\begin{align}
    \tau=\frac{(q)\epsilon}{2} &+\frac{(q+1-\ell)}{2}\log(\epsilon^{q}) + \nonumber \\&+\log\left(2^{(q)\ell/2}\Gamma_{(q)}(\frac{\ell}{2})\right).\label{eq:epsilon}
\end{align}
\end{lm}
\begin{proof} (Lemma \ref{lm:convex})
Denote the eigenvalues of $\sum_{i-1}^\ell g_ig_i^T$ as $\lambda_1,...\lambda_q$. 
The eigenvalues are all non-negative due to the construction of the matrix. 
Note that we can rewrite \eqref{eq:lik} as a new function $\mathfrak{L}_i^j$ in terms of these eigenvalues:
\begin{align}
    \mathcal{L}_q^\ell\left(\sum_{i-1}^\ell g_ig_i^T\right)&=\mathfrak{L}_q^\ell(\lambda_1,...,\lambda_q)\\
    &=\sum_{i=1}^q\frac{(q+1-\ell)}{2}\cdot\log(\lambda_i)+\frac{\lambda_i}{2}+\nonumber\\
    &+\log\left(2^{(q)\ell/2}\Gamma_{(q)}(\frac{\ell}{2})\right).
\end{align}
Furthermore we have that $\mathfrak{L}_q^\ell$ is convex since:
\begin{align}
    \nabla^2\mathfrak{L}_q^\ell(\lambda_1,...,\lambda_q)=\begin{bmatrix}\frac{(\ell-1-q)}{2\lambda_1^2} & 0 & 0\\0 &\ddots&0\\0&0&\frac{\ell-1-q}{2\lambda_q^2}\end{bmatrix}
\end{align}
is positive definite for $\lambda_i>0$. Also note that the function achieves a global minimum at $\lambda_i=\ell-1-q~i=1,...,q$ since:
\begin{align}
    \nabla\mathfrak{L}_q^\ell(\lambda_1,...,\lambda_q)=\begin{bmatrix}\frac{q+1-\ell}{2\lambda_1}+\frac{1}{2}\\\vdots\\
    \frac{q+1-\ell}{2\lambda_q}+\frac{1}{2}\end{bmatrix}
\end{align}
is zero at this point.
If we consider the particular case where $\lambda_1=...=\lambda_q=\epsilon$. Then $(\epsilon,...,\epsilon)$ is a boundary point to the $\tau$ level set of $\mathfrak{L}_q^\ell$. Furthermore we have that the derivative at that point is:
\begin{align}
    \nabla\mathfrak{L}_q^\ell(\epsilon,...,\epsilon)=\begin{bmatrix}\frac{q+1-\ell}{2\epsilon}+\frac{1}{2}\\\vdots\\
    \frac{q+1-\ell}{2\epsilon}+\frac{1}{2}\end{bmatrix}
\end{align}
which is some positive scalar times the vector $\begin{bmatrix}1&\hdots&1\end{bmatrix}^T$. Since the tangent plane at this point is a supporting hyperplane to the $\tau$ sublevel set of $\mathfrak{L}_q^\ell$ we then have that:
\begin{align}
\sum_{i=1}^\ell g_i^Tg_i=\sum_{i=1}^q\lambda_i<\epsilon q
\end{align}
for all $g_i$ such that $\mathcal{L}_q^\ell(\sum_{i=1}^\ell g_ig_i^T)<\tau$.
\end{proof}
Now we return to the prove the theorem

\begin{proof} (Theorem \ref{thm:SHREYAS_LABELED_THIS})
For a given $R\in\Omega_n^{\tau_\mathcal{D}}$:
\begin{align}
    a_{\mathcal{D}}(R_i)=\mathcal{L}_{(m+q)}^\ell\left(\sum_{j=i-\ell+1}^i\psi_j\psi_j^T\right)<\tau_\mathcal{D}
\end{align}
for $i=\ell+k^\prime,...,n$. Lemma \ref{lm:convex} then gives us that:
\begin{align}
    \sum_{j=i-\ell+1}^i\psi_j^T\psi_j=\sum_{j=i-\ell+1}^i\bar{r}_j^T\bar{r}_j+\sum_{j=i-\ell+1}^ie_j^Te_j<(m+q)\epsilon
\end{align}
for $i=\ell+k^\prime,...,n$. Furthermore we have that:
\begin{align}
    R_n^TQ_{i,n}^{\tau_\mathcal{D}}R_n=\sum_{j=i-\ell+1}^i\bar{r}_j^T\bar{r}_j<(m+q)\epsilon
\end{align}
for $i=\ell+k^\prime,...,n$. Therefore $R\in\Theta_n^{\tau_\mathcal{D}}$.
\end{proof}}

Now, we construct an outer approximation to $\mathcal{R}_n^{\tau_*}$ using the constraint sets $\Theta_n^{\tau_*}$:
\begin{thm}
\label{thm:sos}
Suppose $\bar{A}_{n-1}$ and $R_{n-1}$ are defined as in \eqref{eq:abar} and \eqref{eq:R} respectively, $\mathcal{R}_n^{\tau_*}$ is the set in \eqref{eq:Rnset}, $\Phi$ is a compact semi-algebraic set such that $\mathcal{R}_n^{\tau_*} \subset \Phi$, $\Theta_{n-1}^{\tau_*}$ is defined based on the choice of detector 
and:
\begin{align}
    H_n^{\tau_*}=\frac{1}{1-c}H
\end{align}
where $H$ and $c$ are the solution to:
\begin{flalign}
    & & \underset{H\in S~c\in\mathbb{R}}{\inf}\hspace*{0.25cm} & \int_{\Phi} \left( \delta^TH\delta+c \right) ~d\delta &&\\
    & & \text{s.t.}\hspace*{0.25cm} & \delta^TH\delta+c\geq 0 && \delta\in\Phi\\
    & & & R_{n-1}^T\bar{A}_{n-1}^TH\bar{A}_{n-1}R_{n-1}+c-1\geq 0 && R_{n-1}\in \Theta_n^{\tau_*}\label{eq:const2}
\end{flalign}
where $S\subset\mathbb{R}^{p\times p}$ is the set of symmetric matrices.
Then $\mathcal{R}_n^{\tau_*}\subseteq\{\delta~|~\delta^TH_n^{\tau_*}\delta\leq1\}$.
\end{thm}
\remove{\begin{proof}
Let $\delta\in\mathcal{R}_n^{\tau_*}$.
Then, from \eqref{eq:Rnset}, we have that there exists an $R_{n-1}\in\Omega_{n-1}^{\tau_*}\subseteq\Theta_n^{\tau_*}$ such that $\delta=\bar{A}_{n-1}R_{n-1}$.
Constraint \eqref{eq:const2} then gives $R_{n-1}^T\bar{A}_{n-1}^TH\bar{A}_{n-1}R_{n-1}+c\geq1$.
Furthermore $c>1$ since $0\in\Omega_{n-1}^{\tau_*}$, so we can rearrange the inequality resulting in $\delta^T\frac{1}{1-c}H\delta=\delta^TH_n^{\tau_*}\delta\leq1$.
\end{proof}}

One can solve the program in Theorem \ref{thm:sos} using the Spotless optimization toolbox \cite{Tobenkin2018} which formulates the problem as a Semi-Definite Program that can be solved using commercial solvers such as MOSEK \cite{mosek}.
This program assumes that we can find a compact semi-algebraic set $\Phi$ that outer approximates $\mathcal{R}_n^{\tau_*}$, which can be done using the following lemma under the specific case that $N=n$:
\add{\setcounter{lm}{1}}\begin{lm}
\label{lm:bound2}
Suppose $N,n\in\mathbb{N}$, such that $N\geq n$ and if applicable, suppose $N>\ell+k^\prime$ if the detector is the Dynamic Watermarking detector. Furthermore suppose $\tau_*\in\mathbb{R}$ such that $\tau_*>0$, $\bar{A}_{n-1}$ and $R_{N-1}$
 are as in \eqref{eq:abar},\eqref{eq:R}. 
 Then there exists a $\eta\in\mathbb{R}$ such that:
 \begin{align}
     \{\delta=[0_{q\times q(N-n)}~\bar{A}_{n-1}]R_{N-1}~|~R_{N-1}\in\Theta_{N-1}^{\tau_*}\}\subset \mathcal{B}_{\eta}.
 \end{align}
 \end{lm}
 \remove{\begin{proof}(Lemma \ref{lm:bound2})
First we show that $\Theta_{N-1}^{\tau_*}$ is bounded. We denote the upper bounds for the norm of elements in $\Theta_{N-1}^{\tau_*}$ as $\sigma^{\tau_*}$, and we use the decomposition of $R_{N-1}=[\bar{r}_{N-1}^T~\hdots~\bar{r}_0^T]^T\in\Theta_{N-1}^{\tau_*}$.  For the $\chi^2$ detector we have that $\sigma^{\tau_{\chi^2}}=\sqrt{N\tau_{\chi^2}}$ since:
\begin{align}
    \|[\bar{r}_{N-1}^T~\hdots~\bar{r}_0^T]^T\|=\sqrt{\sum_{i=0}^{N-1}\bar{r}_{i}^T\bar{r}_i}\leq\sqrt{N\tau_{\chi^2}}.
\end{align}

Similarly for the CUSUM detector we have that $\sigma^{\tau_{C}}=\sqrt{N(\tau_C+\delta)}$ since:
\begin{align}
    \|[\bar{r}_{N-1}^T~\hdots~\bar{r}_0^T]^T\|=\sqrt{\sum_{i=0}^{N-1}\bar{r}_{i}^T\bar{r}_i}\leq\sqrt{N(\tau_{C}+\delta)}.
\end{align}

In the case of the MEWMA detector we have that  $\sigma^{\tau_{M}}=\sqrt{\frac{N\tau_M(2-\beta)}{\beta}}$ since:
\begin{align}\label{eq:MEWMA_Bound1}
    \|G_i\|=\|\beta\bar{r}_i+(1-\beta)G_{i-1}\|\leq\sqrt{\frac{\tau_M\beta}{2-\beta}}
\end{align}
and:
\begin{align}\label{eq:MEWMA_Bound2}
    \beta\|\bar{r}_i\|-(1-\beta)\sqrt{\frac{\tau_M\beta}{2-\beta}}\leq\|\beta\bar{r}_i+(1-\beta)G_{i-1}\|.
\end{align}
 Combining \eqref{eq:MEWMA_Bound1} and \eqref{eq:MEWMA_Bound2} we get:
 \begin{align}
 \|\bar{r}_i\|\leq\sqrt{\frac{\tau_M(2-\beta)}{\beta}}.
 \end{align}
 Then:
 \begin{align}
    \|[\bar{r}_{N-1}^T~\hdots~\bar{r}_0^T]^T\|=\sqrt{\sum_{i=0}^{N-1}\bar{r}_{i}^T\bar{r}_i}\leq\sqrt{\frac{N\tau_M(2-\beta)}{\beta}}.
\end{align}

In the case of the Dynamic Watermarking detector, we have that $\sigma^{\tau_{\mathcal{D}}}=\sqrt{N(m+q)\epsilon}$, where $\epsilon>\ell-1-q-m$ is the solution to \eqref{eq:epsilonD}, since:
\begin{align}
    \|[\bar{r}_{N-1}^T~\hdots~\bar{r}_0^T]^T\|=\sqrt{\sum_{i=0}^{N-1}\bar{r}_{i}^T\bar{r}_i}\leq\sqrt{N(m+q)\epsilon}.
\end{align}
Then, since:
 \begin{align}
     \|[0_{q\times q(N-n)}&~\bar{A}_{n-1}]R_{N-1}\|\leq\nonumber\\
     &\|[0_{q\times q(N-n)}~\bar{A}_{n-1}]\|~\|R_{N-1}\|,
 \end{align}
let $\eta=\|[0_{q\times q(N-n)}~\bar{A}_{n-1}]\|\sigma^{\tau_*}$.
Then:
 \begin{align}
     \{\delta=[0_{q\times q(N-n)}~\bar{A}_{n-1}]R_{N-1}~|~R_{N-1}\in\Theta_{N-1}^{\tau_*}\}\subset \mathcal{B}_{\eta}.
 \end{align}
 
 \end{proof}}

The program in Theorem \ref{thm:sos} gives an upper bound to $\mathcal{R}_n^{\tau_*}$, which we denote by:
\begin{align}
\mathcal{T}_n^{\tau_*}=\{\delta~|~\delta^TH_n^{\tau_*}\delta\leq1\}. \label{eq:T}  
\end{align}
We dilate $\mathcal{T}_n^{\tau_*}$ to obtain an outer approximation  to $\mathcal{R}^{\tau_*}$:

\begin{thm}
\label{thm:dil}
Suppose $\tau_*\in\mathbb{R}$ such that $\tau_*>0$, $\mathcal{R}^{\tau}$ is as  in \eqref{eq:Rset}, $\mathcal{T}_n^{\tau_*}$ is as in \eqref{eq:T}, and:
\begin{align}
    \mathcal{E}_n^{\tau_*}=\mathcal{T}_n^{\tau_*}\oplus \mathcal{B}_\epsilon,
\end{align}
where:
\begin{align}
    \epsilon=\frac{\|A^n\|}{\sqrt{s_1(H_n^{\tau_*})}(1-\|A^n\|)}.
\end{align}
Then $\mathcal{R}^{\tau_*}\subset \mathcal{E}_n^{\tau_*}$.
\end{thm}
\remove{To prove this result we must first consider the lemma:
\begin{lm}
\label{lm:bound1}
Suppose $n,N,h\in\mathbb{N}$ such that $0<n
\leq h\leq N$,  $R=[r_N^T ... r_0^T]^T\in\Theta_N^{\tau_*}$ where $\Theta_N^{\tau_*}$ is defined based on the choice of detector.
Then $R^\prime=[r_{h}^T ... r_{h-n}^T]^T\in\Theta_n^{\tau_*}$.
\end{lm}
\begin{proof}(of Lemma \ref{lm:bound1})
To prove that $R^\prime$ is in $\Theta_n^{\tau_*}$ we show that each of the constraints associated with $\Theta_n^{\tau_*}$ are included as a constraint associated with $\Theta_N^{\tau_*}$ or that there exists a more restrictive constraint in $\Theta_N^{\tau_*}$.
For the $\chi^2$ test we have the inclusion of all constraints since using \eqref{eq:theta1} and \eqref{eq:chi_con_set_quad2} we have:
\begin{align}
    R^{\prime T}Q_{(i,n)}^{\tau_{\chi^2}}R^\prime=R^TQ_{(i+h-n,N)}^{\tau_{\chi^2}}R<1~\forall i=0,...,n.
\end{align}
Similarly for the CUSUM detector we have that using \eqref{eq:theta2} and \eqref{eq:CUSUM_con_set_quad2} we have:
\begin{align}
R^{\prime T}Q_{(i,j,n)}^{\tau_{C}}R^\prime&=R^TQ_{(i+h-n,j,N)}^{\tau_C}R<1\nonumber\\
&\hspace*{1cm}\forall~i=0,...,n \text{ and } j=0,...,i.
\end{align}
For the MEWMA we have that $\Theta_N^{\tau_*}$ has more restrictive constraints since using \eqref{eq:theta3} and \eqref{eq:MEWMA_con_set_quad2} we have:
\begin{align}
    R^{\prime T}Q_{(i,n)}^{\tau_{M}}R^\prime\leq R^TQ_{(i+h-n,N)}^{\tau_{M}}R<1~\forall i=0,...,n.
\end{align}
For The Dynamic Watermarking Detector we have the inclusion of all constraints since for \eqref{eq:theta4} and \eqref{eq:DW_con_set_quad2} we have:
\begin{align}
 R^{\prime T}Q_{(i,n)}^{\tau_\mathcal{D}}R^\prime=R^TQ_{(i+h-n,N)}^{\tau_\mathcal{D}}R<1~\forall i=\ell+k^\prime,...,n.
\end{align}

\end{proof}

Now we return to proving Theorem \ref{thm:dil}.

\begin{proof}(Theorem \ref{thm:dil})
Let $\delta^\prime\in \mathcal{R}^{\tau_*}$, and assume that $\delta^\prime\notin\mathcal{E}_n^{\tau_*}$. Furthermore let:
\begin{align}
    \epsilon_1=\inf\{\|\delta-\delta^\prime\|~|~\delta\in\mathcal{E}_n^{\tau_*}\}.
\end{align}
Now consider that, for a given $N>n$:
\begin{align}
    \mathcal{R}_N^{\tau_*}\subseteq\{\delta~|~\delta=A_{N-1}X,~X\in\Theta_{N-1}^{\tau_*}\}.
\end{align}
Using Minkowski sums we over-approximate this set further as:
\begin{align}
    \mathcal{R}_N^{\tau_*}&\subseteq\{\delta=[\bar{A}_{n-1}~0_{p\times p(N-n)}]R_{N-1}~|~R_{N-1}\in\Theta_{N-1}^{\tau_*}\}\oplus\nonumber\\
   &\oplus\Bigg(\bigoplus_{i=1}^{j}A^{ni}\{\delta=[0_{p\times pi}~\bar{A}_{n-1}~0_{p\times p(N-n-i)}]R_{N-1}~|\nonumber\\
    &\hspace*{4.6cm}R_{N-1}\in\Theta_{N-1}^{\tau_*}\}\Bigg)\oplus\nonumber\\
    &\oplus A^{n(j+1)}\{\delta=[0_{p\times p(N-nj)}~\bar{A}_{h}]R_{N-1}~|\nonumber\\
    &\hspace*{4.6cm}R_{N-1}\in\Theta_{N-1}^{\tau_*}\}.
\end{align}
where $N$ is evenly divisible by $n$, $j+1$ times and $h$ is the remainder. Applying Lemma \ref{lm:bound2} and \ref{lm:bound1}, we have:
\begin{align}
    \mathcal{R}_N^{\tau_*}&\subseteq\{\delta=\bar{A}_{n-1}R_{n-1}~|~R_{n-1}\in\Theta_{n-1}^{\tau_*}\}\oplus\nonumber\\
    &\oplus\left(\bigoplus_{i=1}^{j}A^{ni}\{\delta=\bar{A}_{n-1}R_{n-1}~|~R_{n-1}\in\Theta_{n-1}^{\tau_*}\}\right)\oplus\nonumber\\
    &\oplus \mathcal{B}_{\eta\|A^{n(j+1)}\|}
\end{align}
where $\eta$ is the maximum radius when applying Lemma \ref{lm:bound2} for $h=0,...,n$.
Let $\sigma=\frac{1}{\sqrt{s_1(H_n^{\tau_*})}}$ then:
\begin{align}
    \{\delta=\bar{A}_{n-1}R_{n-1}~|~R_{n-1}\in\Theta_{n-1}^{\tau_*}\}&\subset\mathcal{T}_i^{\tau_*}\\
    =\{\delta~|~\delta^TH_n^{\tau_*}\delta\leq 1\}&\subset \mathcal{B}_\sigma.
\end{align}
This means that:
\begin{align}
    \mathcal{R}_N^{\tau_*}\subseteq\mathcal{T}_n^{\tau_*}\oplus\left(\bigoplus_{i=1}^{j}\mathcal{B}_{\sigma\|A^{ni}\|}\right)\oplus \mathcal{B}_{\eta\|A^{n(j+1)}\|}.
\end{align}
Since the Minkowski sum of balls is a ball with its radius as the sum of the radii, we can increase the outer approximation by allowing the summation to extend towards infinity:
\begin{align}
    \mathcal{R}_N^{\tau_*}\subseteq\mathcal{T}_n^{\tau_*}\oplus \mathcal{B}_\epsilon\oplus \mathcal{B}_{\eta\|A^{n(j+1)}\|},
\end{align}
where:
\begin{align}
    \epsilon=\frac{\sigma\|A^n\|}{(1-\|A^n\|)}\geq\sum_{i=1}^\infty\sigma\|A^{ni}\|.
\end{align}
Since $j+1>\frac{N}{n}$, there exists an $N_2$ such that for $N>N_2$ we have that $\eta\|A^{n(j+1)}\|<\epsilon_1$ which contradicts $\delta\in\mathcal{R}^{\tau_*}$.

\end{proof}}


\subsection{False Alarm Rate}
\label{sec:false}

While decreasing the threshold for a detector decreases the attack capability, it increases the $R_{FA}$. 
As described in the introduction, to compute this false alarm rate, it is typically assumed that the residuals are independent \cite{Murguia2016CUSUMSensors,Umsonst2018}.
Unfortunately, simulated experiments have noted that this assumption results in a consistent error between the expected and simulated results \cite{Murguia2016CUSUMSensors}. 
This is because the residuals are not independent, which can be confirmed by computing the auto correlation of the sequence of residuals when no attack is present:
\begin{align}
\mathds{E}[r_nr_{n-1}^T]&=\mathds{E}[(C\delta_n-z_n)(C\delta_{n-1}-z_{n-1})^T].
\end{align}
By expanding the product and removing uncorrelated terms one can show that:
\begin{align}
\mathds{E}[r_nr_{n-1}^T]=\mathds{E}[C\delta_n\delta_{n-1}^TC^T]-\mathds{E}[C\delta_nz_{n-1}^T].\label{eq:corr2}
\end{align}
By applying \eqref{eq:observe_error} to \eqref{eq:corr2} and once again canceling uncorrelated terms one can show that:
\begin{align}
\mathds{E}[r_nr_{n-1}^T]=C(A+LC)\mathds{E}[\delta_n\delta_n^T] C^T+L\Sigma_z.
\end{align}
In fact, correlation affects the rate of false alarms \cite{Harris1991StatisticalObservations}. 
To compute the threshold with a specified false alarm rate, we first fix a specific threshold for each detector and simulate the behavior of the system. 
By simulating the system for a long enough time, we can estimate the rate of false alarms associated with the fixed threshold. 
By repeating this approach for a range of thresholds, we can then build a lookup table that associates different thresholds with different false alarm rates.
By linearly interpolating between these thresholds, we can select a threshold that achieves a user-specified false alarm rate.

\section{Simulation and Real-World Comparison}
\label{sec:compare}
This section describes a simulated comparison in which the attack capability for a range of false alarm rates is approximated for each detector, and a real-world comparison in which the ability of each detection algorithm to detect particular attacks is explored.

\subsection{Simulation-Based Comparison of Attack Capability}
\label{sec:simp}
To illustrate the trade-off between the rate of false alarms and attack capability, we provide a comparison of each of the detection algorithms using a 2 dimensional model from \cite{Murguia2018OnAttacks}:
\begin{align*}
    A&=\begin{bmatrix}0.84&0.23\\-0.47&0.12 \end{bmatrix}& 
    ~B&=\begin{bmatrix}0.07& -0.32\\0.23&0.58\end{bmatrix}\\
    C&=\begin{bmatrix}1&0\\2&1\end{bmatrix}&
    ~K&=\begin{bmatrix}1.404&-1.042\\1.842&1.008\end{bmatrix}\\
    L&=\begin{bmatrix}0.0276&0.0448\\-0.01998&-0.0290\end{bmatrix}&
    ~\Sigma_z&=\begin{bmatrix}2&0\\0&2\end{bmatrix}\\
    \Sigma_w&=\begin{bmatrix}0.035&-0.011\\-0.011&0.02\end{bmatrix}&
    ~\Sigma_r&=\begin{bmatrix}2.086&0.134\\0.134&2.230\end{bmatrix}
\end{align*}
with the addition of a watermark with covariance $\Sigma_e=10^{-2}I$. Thresholds for the false alarm rates between $0.01$ and $0.3$ were found by running the simulation under no attack for $10^6$ time steps. 
Using these values, the reachable sets at time step $n=12$ were outer approximated using the optimization program stated in Theorem \ref{thm:sos} and dilated as stated in Theorem \ref{thm:dil} to provide outer approximations of the steady state reachable sets.
The resulting approximations for the $V_{RS}$, defined in \eqref{eq:VRS}, are plotted against the false alarm rate in Figure \ref{fig:reachcomp} for each of the detectors using various detector specific parameter selections. 

\begin{figure}[t]
    \centering
    \includegraphics[trim={2.0in 3.6in 2.2in 3.8in},clip,width=0.44\textwidth]{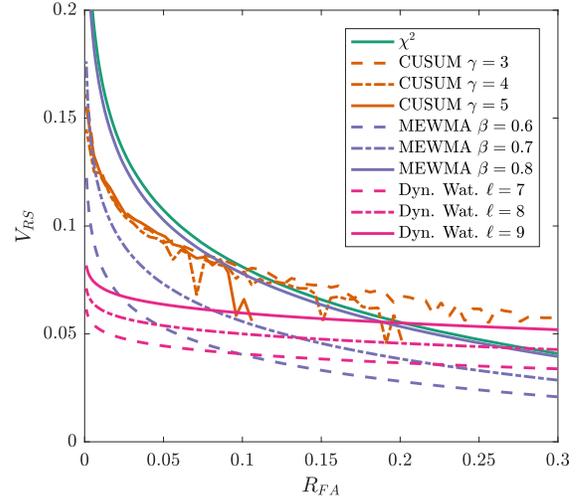}
    \caption{Approximate reachable set volume for varying false alarm rates for the example system in section \ref{sec:simp}}
    \label{fig:reachcomp}
    \vspace*{-0.5cm}
\end{figure}

One may note that while these methods provide smooth curves for the $\chi^2$, MEWMA, and Dynamic Watermarking  detectors, the curves for the CUSUM detector appears discontinuous, and do not span the entire range of $R_{FA}$ values.
The apparent discontinuity is attributed to the fact that, for the $\chi^2$, MEWMA, and Dynamic Watermarking detectors, increasing the threshold $\tau_*$ results in a proportional scaling of the outer approximation of $\mathcal{R}_n^{\tau_*}$. 
However, For the CUSUM detector, changing the threshold does not have this affect. In fact, changing the threshold for the CUSUM detector alters the shape of $\Theta_n^{\tau_C}$, resulting in  the outer approximation of $\mathcal{R}_n^{\tau_C}$ being less conservative for certain threshold values. Furthermore, the shortened span of the curves for the CUSUM detector are a result of certain $R_{FA}$ values being un-achievable for a given forgetting factor.

To determine whether the outer approximation is tight for the $\chi^2$, CUSUM, and MEWMA detectors, simulations were run for 60 $R_{FA}$ values uniformly spaced between 0.01 and 0.3. 
In these simulations, the portion of the observer related to the residual was propagated forward for $10^5$ steps using the dynamics \eqref{eq:errdyn2}. 
The residuals were sampled from a normal distribution with 0 mean and covariance $5I$ and scaled if necessary, to avoid alarms. 
The area of the convex hull of the observer error for the entire simulation was calculated. 
The difference between the over approximated area and the simulated area ranged from $0.0156-0.1408$ for the $\chi^2$, $0.0032-0.1143$ for the CUSUM, and $0.0075-0.1301$ for the MEWMA. 
The results indicate that the attack capability under the Dynamic Watermarking detector is comparable to the classic anomaly detectors as a function of false alarm rate.

\subsection{Real World Implementation}
\label{sec:exp}
\remove{\begin{table*}[t]
    \centering
    \begin{tabular}{|c|c|c|c|c|c|}
     \hline
     \multirow{2}{*}{$R_{FA}$} & \multirow{2}{*}{$\chi^2$} & CUSUM & MEWMA & Dyn. Wat. \\
     & & $\gamma=(15~/~17~/~19)$ & $\beta=(0.6~/~0.7~/~0.8)$ &$\ell=(20~/~25~/~30)$ \\\hline
      \rowcolor{gray!30} 0.05 & 12.08 & ( 51.04 / 12.27 / 2.21 ) & ( 13.09 / 15.00 / 16.85 ) & ( 99.570 / 103.74 / 105.59 ) \\\hline
     \rowcolor{gray!15}0.03 & 14.29 & ( 655.81 / 384.73 / 158.68 ) & ( 15.05 / 17.59 / 20.17 ) & ( 103.05 / 106.73 / 108.58 ) \\\hline
     0.01 & 21.38 & ( 1535.94 / 1445.58 / 1355.22 ) & ( 20.68 / 24.85 / 28.88 ) & ( 108.58 / 110.79 / 113.94 ) \\\hline
\end{tabular}
    \caption[caption]{Experimentally Found Thresholds for Various False Alarm Rates and\\  Detector Specific Parameters for the Real World Implementation in Section \ref{sec:exp}}
    \label{tab:thresh_exp}
    \vspace*{-2em}
\end{table*}}
In this section, we evaluate the ability of each of the anomaly detection schemes to detect attacks using a Segway Robotics Mobility Platform performing a path-following task.
In addition, we illustrate that adding a watermark to the system leads to an imperceptible reduction 
in performance, while significantly improving the detectability of an attack that was missed by classic anomaly detectors. 
Localization was provided by Google Cartographer \cite{Hess2016} using planar lidar and wheel odometry measurements. 
For the purpose of control, a LTV model was fit to the observed data yielding:
\begin{align}
\label{eq:error}
\begin{bmatrix}
e_{\ell,n+1}\\
e_{s,n+1}\\
e_{\theta,n+1}\\
e_{v,n+1}\\
e_{\dot{\theta},n+1}
\end{bmatrix}
=\begin{bmatrix}
e_{\ell,n}+(0.0478\tilde{v}_n)e_{\theta,n}-(0.045\dot{\tilde{\theta}}_n)e_{s,n}\\
e_{s,n}+(0.0478)e_{v,n}+(0.045\dot{\tilde{\theta}}_n)e_{\ell,n}\\
e_{\theta,n}+0.045e_{\dot{\theta},n}\\
e_{v,n}-0.1e_{v,n-4}+0.1u_{v,n}\\
0.6e_{\dot{\theta},n}+0.15e_{\dot{\theta},n-4}+0.24u_{\dot{\theta},n}
\end{bmatrix}
\end{align}
where the state is represented in trajectory error coordinates for a given nominal trajectory where $e_{\ell,n},e_{s,n}, \text{ and } e_{\theta,n}$ are the lateral, longitudinal, and heading error, $e_{v,n},e_{\dot{\theta},n}$ are the error in the velocity and angular velocity, $\tilde{v}_n,\dot{\tilde{\theta}}_n$ are the nominal velocity and angular velocity and $u_{v,n},u_{\dot{\theta},n}$ are the deviation from the nominal inputs. 

For a constant nominal velocity of $0.6$ m/s  and angular velocity of $0$ rad/s, this model can be represented as an LTI model with state vector:
\begin{align}
    x&=\left[\begin{matrix}e_{\ell,n} & e_{s,n} & e_{\theta,n} & e_{v,n} & e_{\dot{\theta},n} & e_{v,n-1}  \end{matrix}\right.\nonumber\\
    &\hspace{1cm}\left.\begin{matrix} e_{v,n-2} & e_{v,n-3} & e_{\dot{\theta},n-1} & e_{\dot{\theta},n-2} & e_{\dot{\theta},n-3}\end{matrix}\right]^T.
\end{align}
For the sake of brevity, the A and B matrices are not stated explicitly but can be found by expanding \eqref{eq:error}. 
The feedback gain matrix $K$  was found to make the closed loop system Schur Stable\add{.}\remove{ and is approximately:
\begin{align}
    K=\begin{bmatrix} 
    0 & -1.639 \\ 0 & -1.984 \\ -0.313 & 0\\ -0.212 & 0 \\ 0 & -0.384 \\ 0.019 & 0\\ 0.020 & 0 \\ 0.021 & 0 \\0.022 & 0 \\0 & -0.039 \\ 0 & -0.043 \\ 0 & -0.052 \\ 0 & -0.065
    \end{bmatrix}
\end{align}}
Similarly the observer gain matrix $L$ was found to make the observer Schur Stable\add{. Both the $K$ and $L$ matrices can be found in the technical report \cite{Porter2018}.}\remove{ and is approximately:
\begin{align}
    L=\begin{bmatrix}
    -0.791	& -0.016 & 0 & 0 & 0 \\
-0.002 & -0.501 & 0 & 0 & -0.022 \\
0 & 0 & -0.272 & -0.025 & 0 \\
0 & 0 & -0.011 & -0.252 & 0 \\
0 & -0.003 & 0 & 0 & -0.240 \\
0 & 0 & -0.013 & -0.258 & 0 \\
0 & 0 & -0.015 & -0.187 & 0 \\
0 & 0 & -0.015 & -0.133	& 0 \\
0 & 0 & -0.014 & -0.091 & 0 \\
0 & -0.004 & 0 & 0 & -0.395 \\
0 & -0.010 & 0 & 0 & -0.145 \\
0 & -0.008 & 0 & 0 & -0.054 \\
0 & -0.005 & 0 & 0 & -0.024 \\
    \end{bmatrix}
\end{align}}
The steady state covariance of the residuals, $\Sigma_r$, was approximated using the sample covariance from data generated by the Segway following a straight line down a 16 m hallway 40 times. 
To avoid the effects of the transient behavior at the start of each run, the beginning of each run was ignored. 
This experiment was repeated a second time after the introduction of a watermark into the control input with covariance:
\begin{align}
    \Sigma_e=\begin{bmatrix}0.02 & 0\\0&0.03\end{bmatrix},
\end{align}
in order to approximate $\Sigma_\psi$.
The average location error was 0.0262 m for the non-watermarked runs and 0.0506 m for the watermarked runs. 
While adding the watermark increased the location error, the average error did not hinder overall performance during the lane-following task.

Threshold values for the false alarm rates of 0.01, 0.03, and 0.05 were approximated for the $\chi^2$, CUSUM, and MEWMA detectors using the residuals from the un-watermarked runs. 
Thresholds for the Dynamic Watermarking detector and the same false alarm rates were found using the residuals from the watermarked runs. 
The resulting threshold values are displayed in \add{the technical report \cite{Porter2018}.}\remove{ Table \ref{tab:thresh_exp}.} 

Two attacks, following differing models, were then applied. 
Attack model 1 assumes that the attacker adds random noise to the system such that $v_n\sim\mathcal{N}(0,10^{-5} I)$.
Attack model 2 takes the form:
\begin{align}
    v_n&=-(Cx_n+z_n) +C\xi_n+\zeta_n.\label{eq:attack}
\end{align}
For this model, the attack measurement noise $\zeta_n$ is added to the false state such that $\zeta_n\sim\mathcal{N}(0,\Sigma_\zeta)$ and the false state $\xi_n\in\mathbb{R}^p$ is updated according to the closed loop dynamics of the system:
\begin{align}
    \xi_{n+1}&=(A+BK)\xi_n+\omega_n
\end{align}
with attack process noise $\omega_n\sim\mathcal{N}(0,\Sigma_\omega)$. 
The attack process and measurement noise were chosen to leave the distribution of the residuals unchanged. 

For each attack, 10 experimental runs were completed without a watermark and 10 with a watermark for a total of 40 experimental runs.  
The runs with a watermark were used in evaluating the Dynamic Watermarking detector, while all other detectors used the un-watermarked data. 
The resulting detection rates, defined as the number of alarms divided by the total number of time steps in the attacked runs, are displayed in Table \ref{tab:exp_exam}.

\begin{table}[ht]
    \centering
    \begin{tabular}{|c|c|c|c|}
     \hline
     \multirow{2}{*}{Method} & \multirow{2}{*}{$R_{FA}$} & Attack Model 1 & Attack Model 2\\
     & & Detection Rates & Detection Rates\\\hline
     \multirow{3}{*}{$\chi^2$} & 0.05 \cellcolor{gray!30}&  0.69 \cellcolor{gray!30}& 0.03 \cellcolor{gray!30}\\\cline{2-4}
      & 0.03 \cellcolor{gray!15}&  0.62 \cellcolor{gray!15}& 0.01 \cellcolor{gray!15}\\\cline{2-4} 
     & 0.01 &  0.47 & 0.00  \\\hline
     \multirow{3}{*}{$\underset{\gamma=(15/17/19)}{\text{CUSUM}}$} &\cellcolor{gray!30}0.05 &\cellcolor{gray!30}( 0.98 / 0.99 / 0.99 ) &\cellcolor{gray!30}( 0.00 / 0.00 / 0.00 )\\\cline{2-4}
     &\cellcolor{gray!15}0.03 &\cellcolor{gray!15}( 0.81 / 0.86 / 0.92 )  &\cellcolor{gray!15}( 0.00 / 0.00 / 0.00 ) \\\cline{2-4} 
     & 0.01 & ( 0.58 / 0.54 / 0.51 ) & ( 0.00 / 0.00 / 0.00 )\\\hline
     \multirow{3}{*}{$\underset{\beta=(0.6/0.7/0.8)}{\text{MEWMA}}$} &\cellcolor{gray!30}0.05 &\cellcolor{gray!30}( 0.51 / 0.57 / 0.62 ) &\cellcolor{gray!30}( 0.03 / 0.03 / 0.03 ) \\\cline{2-4}
     &\cellcolor{gray!15}0.03 &\cellcolor{gray!15}( 0.45 / 0.50 / 0.54 ) &\cellcolor{gray!15}( 0.01 / 0.01 / 0.01 )   \\\cline{2-4} 
     & 0.01 & ( 0.32 / 0.35 / 0.39 ) & ( 0.00 / 0.00 / 0.00 ) \\\hline
     \multirow{3}{*}{$\underset{\ell=(20/25/30)}{\text{Dyn. Wat.}}$} &\cellcolor{gray!30}0.05 &\cellcolor{gray!30}( \textbf{1.00} / \textbf{1.00} / \textbf{1.00} ) &\cellcolor{gray!30}( 0.98 / \textbf{1.00} / \textbf{1.00} )\\\cline{2-4}
     &\cellcolor{gray!15}0.03 &\cellcolor{gray!15}( \textbf{1.00} / \textbf{1.00} / \textbf{1.00} ) &\cellcolor{gray!15}( 0.97 / 0.99 / \textbf{1.00} )\\\cline{2-4} 
     & 0.01 & ( \textbf{1.00} / \textbf{1.00} / \textbf{1.00} ) & ( 0.95 / 0.98 / \textbf{1.00} )\\\hline
\end{tabular}
    \caption{Experimentally Found Alarm Rates For Various Detector Specific Parameters for the Attack Models from Section \ref{sec:exp}}
    \label{tab:exp_exam}
    \vspace*{-0.5cm}
\end{table}

For attack model 1, 
all of the detectors are able to reliably detect the attack, confirming that the implementation of the detectors is correct.
For attack model 2 the detection rate decrease from the $R_{FA}$ for the $\chi^2$, CUSUM, and MEWMA detectors. 
This may be due to the residuals for the un-attacked system not being distributed as a Gaussian distribution resulting in higher threshold values. 
Since attack model 2 replaces the feedback completely, the resulting residuals, when under attack, do follow a Gaussian distribution which then results in lower detection rates. 
The Dynamic Watermarking detector in the presence of the second attack provides a high detection rates for each set of parameters, and in some cases achieves a perfect detection rate. 

\section{Conclusion}
\label{sec:con}

This paper derives a method to evaluate the capability of an attacker without raising an alarm for $\chi^2$, CUSUM, MEWMA, and Dynamic Watermarking detectors.
Using this notion of attack capability, this paper illustrates that all considered detectors have comparable performance.
However, on a real-world system, Dynamic Watermarking is the only detector that is capable of detecting the presence of a certain class of attacks. 


\renewcommand{\bibfont}{\normalfont\small}
{\renewcommand{\markboth}[2]{}
\printbibliography}
\end{document}